\documentclass[final,3p,times,12pt]{elsarticle}

\usepackage{amssymb}
\usepackage{amsthm}
\usepackage{amsmath}
\usepackage{latexsym}
\usepackage{times}
\usepackage{graphicx}
\usepackage{epsfig}
\usepackage{epstopdf}
\usepackage{mathrsfs}

\newtheorem{Theorem}{Theorem}
\newtheorem{Remark}{Remark}
\newtheorem{Definition}{Definition}

\newcommand{\clo}{\overline}
\newcommand{\ra}{\rangle}
\newcommand{\la}{\langle}

\newcommand{\usc}{u^{\rm sc}}
\newcommand{\uinc}{u^{\rm inc}}
\newcommand{\ut}{u^{\rm t}}

\usepackage[active]{srcltx}

\journal{Elsevier}

\begin{document}

\begin{frontmatter}

%% Title, authors and addresses

%% use the tnoteref command within \title for footnotes;
%% use the tnotetext command for theassociated footnote;
%% use the fnref command within \author or \address for footnotes;
%% use the fntext command for theassociated footnote;
%% use the corref command within \author for corresponding author footnotes;
%% use the cortext command for theassociated footnote;
%% use the ead command for the email address,
%% and the form \ead[url] for the home page:
%% \title{Title\tnoteref{label1}}
%% \tnotetext[label1]{}
%% \author{Name\corref{cor1}\fnref{label2}}
%% \ead{email address}
%% \ead[url]{home page}
%% \fntext[label2]{}
%% \cortext[cor1]{}
%% \address{Address\fnref{label3}}
%% \fntext[label3]{}

%% use optional labels to link authors explicitly to addresses:
%% \author[label1,label2]{}
%% \address[label1]{}
%% \address[label2]{}

\title{\textbf{On-surface radiation condition for multiple scattering of waves}}

\author[]{Sebastian Acosta\corref{cor1}}
\ead{sebastian.acosta@rice.edu}
\ead[url]{sites.google.com/site/acostasebastian01}

\address{Computational and Applied Mathematics, Rice University, Houston, TX \\
Department of Pediatric Cardiology, Baylor College of Medicine, Houston, TX }

\cortext[cor1]{Corresponding author}

\begin{abstract}
The formulation of the \textit{on-surface radiation condition} (OSRC) is extended to handle wave scattering problems in the presence of multiple obstacles. The new multiple-OSRC simultaneously accounts for the outgoing behavior of the wave fields, as well as, the multiple wave reflections between the obstacles. Like boundary integral equations (BIE), this method leads to a reduction in dimensionality (from volume to surface) of the discretization region. However, as opposed to BIE, the proposed technique leads to boundary integral equations with \textit{smooth} kernels. Hence, these Fredholm integral equations can be handled accurately and robustly with standard numerical approaches without the need to remove singularities. Moreover, under weak scattering conditions, this approach renders a convergent iterative method which bypasses the need to solve \textit{single scattering} problems at each iteration.

Inherited from the original OSRC, the proposed multiple-OSRC is generally a crude approximate method. If accuracy is not satisfactory, this approach may serve as a good initial guess or as an inexpensive pre-conditioner for Krylov iterative solutions of BIE.
\end{abstract}

\begin{keyword}
%% keywords here, in the form: keyword \sep keyword
Multiple obstacles \sep wave scattering \sep on surface radiation condition \sep absorbing boundary condition \sep Helmholtz equation

%% PACS codes here, in the form: \PACS code \sep code
%% MSC codes here, in the form: \MSC code \sep code
%% or \MSC[2008] code \sep code (2000 is the default)
\end{keyword}

\end{frontmatter}

%\linenumbers

%%%%%%%%%%%%%%%%%%%%%%%%%%%%%%%%%%%%%%%%%%%%%%%%%%%%%%%%%%%%%%%%%%%%
%%   S E C T I O N
%%%%%%%%%%%%%%%%%%%%%%%%%%%%%%%%%%%%%%%%%%%%%%%%%%%%%%%%%%%%%%%%%%%%
\section{Introduction} \label{Section.Intro}

Multiple scattering problems emerge in many applications dealing with wave phenomena in acoustics, electromagnetism, elastodynamics or hydrodynamics. We restrict our attention to problems governed by the Helmholtz equation which is one of the most useful models for such wave phenomena. However, analogous ideas can be explored for the Maxwell system and the equations of elasticity. We consider a finite number of disjoint impenetrable obstacles embedded in a homogeneous isotropic medium. A given incident wave impinges upon the obstacles, and the problem amounts to calculate the scattered field. From the computational point of view, this is a very challenging problem. Difficulties are encountered due to the unboundedness of the medium, the appropriate satisfaction of the outgoing radiation condition at infinity, the number of obstacles and their location, their unavoidable interaction, and the geometrical characteristics of each one of them. An excellent resource for the study of these problems is the book by Martin \cite{MartinBook}.

When the obstacle's boundary conforms to simple shapes such as spheres or ellipsoids, then methods based on separation of variables and Fourier expansions render a good approach \cite{MartinBook,Antoine2008,Gau-Hua-Str-1995,Lee-2012}. Otherwise, one is generally forced to consider numerical methods based on discretization. These methods typically belong to two main categories: volume and surface discretizations. 

In the first category, we find the finite element (FEM) and the finite difference (FDM) methods among others. FEM and FDM lead to sparse matrix problems. However, their major drawback is the appropriate handling of the unboundedness of the medium. Many nonreflecting boundary conditions have been produced to truncate the physical domain while ensuring the outgoing behavior of the scattered wave. Some of these absorbing boundary conditions include local conditions, global conditions, infinite elements, and perfectly matched layers. These techniques have been extensively reviewed in the following articles and books \cite{GivoliReview2,Tsynkov1998,BHPR2010,GivoliBook,Ihlenburg}.

The second category is based on boundary integral equations (BIE) leading to surface discretizations. Here we find the boundary element method (BEM), the null-field equations and the T-matrix method. See \cite{MartinBook} for an excellent review. They enjoys a reduction in dimensionality (from volume to surface) of the discretization region and the automatic satisfaction of the radiation condition at infinity. Moreover, the formulation of BIE is indifferent to the number of obstacles making it well suited for multiple scattering. However, BIE may become quite costly since they lead to singular integral kernels, and dense and possibly ill-conditioned matrices. Another technique leading to surface discretization is the method of \textit{on-surface radiation conditions} (OSRC) originally developed by Kriegsmann, Taflove and Umashankar \cite{Kriegsmann87}. It is an approximate method that applies local absorbing conditions directly on the surface of the obstacle. As originally proposed, in the presence of a single obstacle, the OSRC leads to an explicit formula (up to numerical integration) avoiding the costly linear system obtained from BEM. However, the OSRC approximation is generally  crude since the application of a local absorbing condition directly on the surface renders a low order method independent of mesh refinement. It is understood in the pertaining literature that the OSRC is designed to sacrifice accuracy in favor of computational speed. Therefore, the OSRC should not be employed if a high degree of accuracy is needed. For studies on high-order absorbing boundary conditions, we refer the reader to \cite{GivoliReview2,Givoli2001,Rab-Giv-Bec-2010} and references therein.

Until recently, all of the volume discretization methods for scattering problems were exclusively formulated to handle a single scatterer. If several obstacles are present, the common practice is to ignore the multiple-component nature of the scatterer and enclose all obstacles with a sufficiently large artificial boundary. This has changed through the work of Grote, Kirsch and Sim \cite{Grote01,Grote02,Grote03} who formulated appropriate nonreflecting boundary conditions on artificial sub-boundaries each one enclosing a separate obstacle. An analogous work has been carried out by Jiang and Zheng \cite{JiangZheng} for the PML approach. Similarly, the main objective of the present Article is to derive an OSRC for multiple scattering problems which brings the absorbing condition to the surface of each obstacle. The naive implementation of the original OSRC to multiple scattering leads to an erroneous solution which neglects the unavoidable interaction between the obstacles. Therefore, we propose a \textit{multiple-OSRC} formulated to simultaneously account for the outgoing behavior of the wave fields \textit{and} the multiple reflections between the obstacles. We may summarize the advantages of the proposed method as follows.

\begin{itemize}
 \item[1.] In common with BIE, the multiple-OSRC leads to a reduction in dimensionality (from volume to surface) of the discretization region and the automatic satisfaction of the radiation condition at infinity.
 
 \item[2.] As opposed to BIE, the multiple-OSRC leads to integral equations with \textit{smooth kernels}. These Fredholm integral equations of the second kind can be handled accurately and robustly with standard numerical approaches such as Nystr\"{o}m, Galerkin and collocation methods, \textit{without} the need to remove singularities. In addition, integral operators with analytic kernels can be approximated with spectrally accurate low-rank operators using separable expansions of the kernels. See for instance \cite[Thm. 2]{Mart-Rok-2007} or \cite[Ch. 2-3]{MartinBook}.
 
 \item[3.] Under appropriate conditions, the multiple-OSRC leads to convergent Neumann series or so-called \textit{orders of scattering}. However, as opposed to BIE, at each iteration there is no need to solve a single-scattering problem for each obstacle because the OSRC renders an explicit formula. Thus, after discretization, there is no need to build or invert a matrix. See details in Section \ref{Section.Iterative}.
 
 \item[4.] If the accuracy of the proposed method is not satisfactory, the multiple-OSRC approach may serve as an inexpensive pre-conditioner for Krylov iterative solutions of BIE \cite{Ant-Dar-2005,Ant-Dar-2007},
or as an extraordinarily good initial guess for such iterative techniques. The multiple-OSRC may also be employed as a fast approximate method to explore parameter-spaces for optimization procedures and the construction of reduced bases \cite{Gan-Hes-Sta-2012}.
\end{itemize}

The proposed multiple-OSRC has a major drawback inherited from the original OSRC. In contrast with BIE, the multiple-OSRC is only an approximate method whose accuracy cannot be improved by grid refinement and no error bounds are available. In many cases, as seen in the literature \cite{Kriegsmann87,Ant-InBook2008,Antoine01,Antoine06,Engquist02}, the OSRC method leads to crude approximations of the exact solution. In addition, the OSRC performance deteriorates even more for non-convex obstacles, or rapid changes in the boundary curvature including corners or edges.
  
However, impressive advances have been accomplished by Antoine, Barucq, Bendali and others with the incorporation of surface curvature information into the absorbing condition. See \cite{Ant-InBook2008,Antoine01,Antoine99,AntBar01,BarDiazDup12} and references therein. Using pseudo-differential calculus, they perform a careful derivation of local absorbing conditions for surfaces of arbitrary shape. See also \cite{Alber-1981} for the use of geometrical optics for non-convex obstacles.

%%%%%%%%%%%%%%%%%%%%%%%%%%%%%%%%%%%%%%%%%%%%%%%%%%%%%%%%%%%%%%%%%%%%
%%   S E C T I O N
%%%%%%%%%%%%%%%%%%%%%%%%%%%%%%%%%%%%%%%%%%%%%%%%%%%%%%%%%%%%%%%%%%%%
\section{Formulation of the problem} \label{Section.Formulation}

In this section we setup the mathematical formulation of the scattering problem and review the well-known Green's integral representation of the solution. We shall work on three-dimensional scenarios, but the formulation in Sections \ref{Section.Formulation}-\ref{Section.AdjointDtN} is also valid in two dimensions with obvious modifications. 

The scatterer consists of $J$ disjoint obstacles each occupying a simply connected bounded domain $\Omega^{-}_{j}$ with respective smooth boundary $\Gamma_{j}$ for $j=1,2,...,J$. The open region external to $\Gamma_{j}$ is denoted by $\Omega^{+}_{j}$. We also define 
\begin{eqnarray*}
 \Omega^{-} = \bigcup_{j=1}^{J} \Omega^{-}_{j},  \quad \Omega^{+} = \bigcap_{j=1}^{J} \Omega^{+}_{j} \quad \text{and} \quad \Gamma = \bigcup_{j=1}^{J} \Gamma_{j}.
\end{eqnarray*}

As usual, an incident wave $\uinc$ impinges upon the $J$ obstacles. The total field $\ut$ is
decomposed as $\ut = \uinc + \usc$ in $\clo{\Omega^{+}}$ where $\usc$ represents the
scattered field induced by the presence of the obstacles. For simplicity we assume a Dirichlet condition $\ut = 0$ on $\Gamma$ so that the outgoing scattered field $\usc$ satisfies, 
\begin{eqnarray}
&& \Delta \usc + k^2 \usc = 0 \quad\qquad \text{in $\Omega^{+}$}, \label{BVPsc1} \\
&& \usc = - \uinc \qquad\qquad \text{on $\Gamma$,} \label{BVPsc2} \\
&& \lim_{r \rightarrow \infty} r \left( \frac{\partial \usc} {\partial r} - i k \usc \right) = 0. \label{BVPsc3}
\end{eqnarray}
The limit in (\ref{BVPsc3}) is known as the Sommerfeld radiation condition where $r = |x|$ for $x \in \mathbb{R}^3$. The well-posedness of the BVP (\ref{BVPsc1})-(\ref{BVPsc3}) is studied in \cite{ColtonKress02,Nedelec01,McLean2000}. In the derivation of the multiple-OSRC we shall rely on the following Green's integral representation of the scattered field \cite{ColtonKress02,Nedelec01,McLean2000}. The scattered wave $\usc$ satisfies the following identity,
\begin{eqnarray}
\usc(x) = \int_{\Gamma} \bigg[ \usc(y) \frac{\partial \Phi(x,y)}{\partial n(y)} - \frac{\partial \usc}{\partial n} (y) \Phi(x,y) \bigg] dS(y), \quad x \in \Omega^{+}, \label{GreenId01}
\end{eqnarray}
where $n$ denotes the unit normal vector to the boundary $\Gamma$ directed into $\Omega^{+}$ and $\Phi$ is the outgoing fundamental solution for the Helmholtz equation given by
\begin{eqnarray}
\Phi(x,y) = \frac{1}{4 \pi} \frac{e^{i k |x - y|}}{|x - y|}, \quad x \neq y. \label{Fund01}
\end{eqnarray}

The BVP (\ref{BVPsc1})-(\ref{BVPsc3}) is well-posed in appropriate spaces. So one may define the exterior Dirichlet-to-Neumann (DtN) operator $\Lambda$ mapping $\usc |_{\Gamma}\mapsto \partial_{n} \usc |_{\Gamma}$. Since $\Gamma$ is the union of several sub-boundaries $\Gamma_{j}$, we will refer to $\Lambda$ as the \textit{multiple-DtN} map. In view of the Green's identity (\ref{GreenId01}) and the Dirichlet boundary condition (\ref{BVPsc2}), the knowledge of the multiple-DtN map $\Lambda$ renders an explicit formula for the scattered field in terms of the incident field,
\begin{eqnarray}
\usc(x) = - \int_{\Gamma} \bigg[ \uinc(y) \frac{\partial \Phi(x,y)}{\partial n(y)} - (\Lambda \uinc)(y) \Phi(x,y) \bigg] dS(y), \quad x \in \Omega^{+}, \label{Explicit01}
\end{eqnarray}
Of course, the application of the multiple-DtN operator $\Lambda$ amounts to solve the full multiple scattering problem in the first place. Hence, expression (\ref{Explicit01}) has little practical use in its current form.

Many explicit approximations of the DtN operator have been developed in the past. They can be derived using pseudo-differential calculus and analytic expansions. For instance, one may use the well-known Engquist-Majda \cite{Engquist01}, Bayliss-Gunzburger-Turkel \cite{Bayliss01} or Higdon \cite{Higdon01} conditions to find applicable approximations of the DtN operator. For details, see the work of Antoine \textit{et al.} \cite{Antoine99,Antoine01,Antoine06}, Atle and Engquist \cite{Engquist02} and references therein. The main idea in the original formulation of the OSRC \cite{Kriegsmann87} is to simply replace the DtN operator in (\ref{Explicit01}) by one of its approximations, and use the resulting formula as an explicit approximation for the scattered field. 

Unfortunately, these approximations are based on the assumption that the scattered field radiates from a single boundary in the outgoing direction. This is the case when $\Omega^{-}$ consists of a single connected component. However, for the multiple scattering problem, the scatterer $\Omega^{-}$ is disconnected since it is the union of several disjoint sub-scatterers $\Omega^{-}_{j}$ for $j=1,2,...,J$. As a result, the scattered field bounces back and forth between obstacles inducing a complicated reflection pattern. As pointed out in the Introduction, the naive implementation of the original OSRC to multiple scattering problems leads to an erroneous solution which neglects the unavoidable interaction between the obstacles. In the remainder of this Article, we shall derive a multiple-OSRC which simultaneously accounts for the outgoing behavior of the scattered field and the multiple reflections between the obstacles.

%%%%%%%%%%%%%%%%%%%%%%%%%%%%%%%%%%%%%%%%%%%%%%%%%%%%%%%%%%%%%%%%%%%%
%%   S E C T I O N
%%%%%%%%%%%%%%%%%%%%%%%%%%%%%%%%%%%%%%%%%%%%%%%%%%%%%%%%%%%%%%%%%%%%

\section{Derivation of the multiple-OSRC} \label{Section.Multiple-OSRC}

Our derivation of the multiple-OSRC rests upon the following fundamental decomposition theorem for multiple scattering problems. This theorem was explicitly proven in \cite{Grote01} and \cite{JiangZheng} for the two- and three-dimensional settings, respectively. In both cases, it was assumed that the obstacles are well-separated and conform to canonical shapes. Another proof was presented by Balabane \cite{Balabane04} without assumptions on the shape or distance between of the obstacles as long as they are disjoint. See also \cite{Wang-Liu-2013,JCP2010,CoaJoly2012,Aco-Vill-Mal} for alternative proofs and similar applications of this decomposition theorem.
\begin{Theorem} \label{TheoremDecomp}
Let $\usc$ be the solution to the BVP (\ref{BVPsc1})-(\ref{BVPsc3}). Then, $\usc$ can be uniquely
decomposed into purely outgoing wave fields $u_{j}$ for $j=1,2,...,J$ such
that
\begin{equation}
\usc =  \sum_{j=1}^{J} u_{j},\qquad \text{in
$\overline{\Omega^{+}}$}, \label{Decomp1}
\end{equation}
where $u_{j}$ radiates purely from $\Gamma_{j}$, that is,
\begin{eqnarray}
\Delta u_{j} + k^2 u_{j} = 0 \quad \text{in $\Omega_{j}^{+}$}, \qquad \text{and} \qquad \lim_{r \rightarrow \infty} r \, \bigg(\frac{\partial u_{j}}{\partial r} - i k u_{j} \bigg) = 0, \label{Decomp2}
\end{eqnarray}
\end{Theorem}

First of all, notice from (\ref{Decomp2}) that the purely outgoing field $u_{j}$ is a radiating solution to the Helmholtz equation on all of $\Omega_{j}^{+}$, including the interior of the other obstacles $\Omega_{i}^{-}$ for $i \neq j$. This is precisely what we mean by a \textit{purely outgoing} field with respect to the radiating boundary $\Gamma_{j}$. Also notice that the purely outgoing field $u_{j}$ is completely determined by its Dirichlet boundary data on $\Gamma_{j}$. So once this data is fixed, the purely outgoing field $u_{j}$ is completely oblivious to the presence of the other obstacles. 

Now, we turn our attention to the following problem: Find the Dirichlet boundary data of each purely outgoing field $v_{j}$ on its respective radiating boundary $\Gamma_{j}$. Again, we rely on the Green's integral representation to obtain,
\begin{eqnarray}
u_{j}(x) = \int_{\Gamma_{j}} \bigg[ u_{j}(y) \frac{\partial \Phi(x,y)}{\partial n_{j}(y)} - \frac{\partial u_{j}}{\partial n_{j}} (y) \Phi(x,y) \bigg] dS(y), \quad x \in \Omega_{j}^{+}, \quad \text{for all $j=1,2...,J$.} \label{GreenIdV01}
\end{eqnarray}
Here we pause to emphasize the fundamental difference between (\ref{GreenIdV01}) and (\ref{GreenId01}). Notice that in representation (\ref{GreenIdV01}), we only integrate over a single boundary $\Gamma_{j}$ from which $u_{j}$ radiates. Therefore, from the well-posedness of exterior Dirichlet problems, we may define a Dirichlet-to-Neumann operator $\Lambda_{j}$ mapping $u_{j}|_{\Gamma_{j}} \mapsto \partial_{n_{j}} u_{j} |_{\Gamma_{j}}$. We will refer to $\Lambda_{j}$ as the \textit{single-DtN} map since it is associated only with a single sub-boundary $\Gamma_{j}$.
The knowledge of the single-DtN map $\Lambda_{j}$ renders an explicit formula for this purely outgoing field $u_{j}$ in terms of its Dirichlet data,
\begin{eqnarray}
u_{j}(x) = \int_{\Gamma_{j}} \bigg[ u_{j}(y) \frac{\partial \Phi(x,y)}{\partial n_{j}(y)} - (\Lambda_{j} u_{j})(y) \Phi(x,y) \bigg] dS(y), \quad x \in \Omega_{j}^{+}, \quad \text{for all $j=1,2...,J$.} \label{ExplicitV01}
\end{eqnarray}
As opposed to (\ref{Explicit01}), the expression above now has great practical value since the single-DtN map $\Lambda_{j}$ is associated with a \textit{single-radiation problem} and it can be explicitly approximated by the standard absorbing boundary conditions \cite{GivoliReview2,Tsynkov1998,GivoliBook,Ihlenburg,Kriegsmann87,Engquist01,Bayliss01,Higdon01}.

Keeping in mind that each single-DtN map $\Lambda_{j}$ can be suitably approximated by a single-OSRC, then one can setup a system of linear integral equations for the Dirichlet data of each purely outgoing field $u_{j}$. This is accomplished by enforcing the decomposition (\ref{Decomp1}) on $\Gamma = \cup_{j=1}^{J} \Gamma_{j}$ in combination with Dirichlet boundary condition (\ref{BVPsc2}) for $\usc$. This leads to the following system,
\begin{eqnarray}
u_{j}(x) + \sum_{i \neq j} \int_{\Gamma_{i}} \bigg[ u_{i}(y) \frac{\partial \Phi(x,y)}{\partial n_{i}(y)} - (\Lambda_{i} u_{i})(y) \Phi(x,y) \bigg] dS(y) = - \uinc(x), \quad x \in \Gamma_{j}, \label{System01}
\end{eqnarray}
for all $j=1,2...,J$. Conceptually, the system (\ref{System01}) represents the proposed \textit{multiple-OSRC}. 

In order to simplify the derivation of the multiple-OSRC, we have purposely avoided the specification of normed spaces to which the wave fields belong. However, at this point we begin to setup the multiple-OSRC in operator notation and analyze its properties. Hence, it is convenient to state spaces and norms with precision. From the derivation above, we are lead to consider spaces on which the single-DtN operator $\Lambda_{j}$ is bounded. We have the choice of the classical H\"{o}lder spaces \cite[Ch. 3]{ColtonKress02} or the Sobolev space setting \cite[Ch. 4]{McLean2000}. We have chosen the latter. From the well-posedness of the weak formulation of exterior Dirichlet radiating problems \cite{McLean2000}, we have the following regularity properties for $s \in \mathbb{R}$,
\begin{eqnarray*}
&& \text{Incident field} \quad \uinc \in H_{\rm loc}^{s+1/2}(\mathbb{R}^3) \\
&& \text{Scattered field} \quad \usc \in H_{\rm loc}^{s+1/2}(\Omega^{+}) \\
&& \text{Purely outgoing field} \quad u_{j} \in H_{\rm loc}^{s+1/2}(\Omega_{j}^{+}) \\
&& \text{Trace operator for $\Omega_{j}^{+}$ denoted by} \quad \gamma_{j} : H_{\rm loc}^{s+1/2}(\Omega_{j}^{+}) \to H^{s}(\Gamma_{j}) \\
&& \text{Single-DtN  operator for $\Gamma_{j}$ denoted by} \quad \Lambda_{j} : H^{s}(\Gamma_{j}) \to H^{s-1}(\Gamma_{j})
\end{eqnarray*}
We also define the inner product on $H^{0}(\Gamma_{j})$ by
\begin{eqnarray}
\la w , v \ra_{j} = \int_{\Gamma_{j}} \clo{w}(y) v(y) \, dS(y), \label{InnerProd}
\end{eqnarray}
which is also generalized to coincide with the duality pairing between a functional $w \in H^{-s}(\Gamma_{j})$ and a vector $v \in H^{s}(\Gamma_{j})$ for $s \geq 0$. We also generalize the complex conjugate $\clo{w} \in H^{-s}(\Gamma_{j})$ of the functional $w \in H^{-s}(\Gamma_{j})$ by $\la \clo{w} , v \ra_{j} = \clo{\la w , \clo{v} \ra}_{j}$ for all $v \in H^{s}(\Gamma_{j})$.

Now we define the wave \textit{propagation} operators $P_{ij} : H^{s}(\Gamma_{j}) \to H^{s}(\Gamma_{i})$ given by
\begin{eqnarray}
(P_{ij} v)(x) := \la \partial_{n} \clo{\Phi}(x,\cdot) , v \ra_{j} - \la \clo{\Lambda_{j} v} , \Phi(x,\cdot) \ra_{j}, \qquad x \in \Gamma_{i}, \qquad i \neq j. \label{PropagP} 
\end{eqnarray}
So the operator $P_{ij}$ represents the propagation of the wave field from surface $\Gamma_{j}$ to surface $\Gamma_{i}$.

It is now clear that the system of equations (\ref{System01}) can be written in operator notation as follows,
\begin{eqnarray}
\left[ \begin{array}{cccc}
          I    &  P_{1,2} & \cdots & P_{1,J} \\
       P_{2,1} &     I    & \cdots & P_{2,J} \\
        \vdots &  \vdots  & \ddots & \vdots  \\
       P_{J,1} & P_{J,2}  & \cdots &   I 
\end{array} \right] 
\left[ \begin{array}{c}
 \gamma_{1} u_{1} \\
 \gamma_{2} u_{2} \\
 \vdots \\
 \gamma_{J} u_{J}
       \end{array} \right] 
=
- \left[ \begin{array}{c}
 \gamma_{1} \uinc \\
 \gamma_{2} \uinc \\
 \vdots \\
 \gamma_{J} \uinc
       \end{array} \right]. \label{SystemOper01}
\end{eqnarray}
Defining $P : \prod_{j=1}^{J} H^{s}(\Gamma_{j}) \to \prod_{j=1}^{J} H^{s}(\Gamma_{j})$ given by 
\begin{eqnarray}
P = \left[ \begin{array}{cccc}
          0    &  P_{1,2} & \cdots & P_{1,J} \\
       P_{2,1} &     0    & \cdots & P_{2,J} \\
        \vdots &  \vdots  & \ddots & \vdots  \\
       P_{J,1} & P_{J,2}  & \cdots &   0 
\end{array} \right], \label{BigPropagP}
\end{eqnarray}
we can express (\ref{SystemOper01}) in compressed notation as
\begin{eqnarray}
(I + P) u = f, \label{SystemOper02}
\end{eqnarray}
where $u = (\gamma_{1}u_{1},...,\gamma_{J}u_{J}) \in \prod_{j=1}^{J} H^{s}(\Gamma_{j})$ and $f = - (\gamma_{1}\uinc,...,\gamma_{J}\uinc) \in \prod_{j=1}^{J} H^{s}(\Gamma_{j})$. The product space $\prod_{j=1}^{J} H^{s}(\Gamma_{j})$ is made a Banach space when equipped with the Sobolev norm on each product space $H^{s}(\Gamma_{j})$ composed with any norm on $\mathbb{C}^J$. So we equip $\prod_{j=1}^{J} H^{s}(\Gamma_{j})$ with the following norm,
\begin{eqnarray}
\| u \| := \max_{j=1,...,J} \| u_{j} \|_{H^{s}(\Gamma_{j})}, \quad \text{where $u = (u_{1},...,u_{J}) \in \prod_{j=1}^{J} H^{s}(\Gamma_{j})$} \label{Norm}
\end{eqnarray}

The well-posedness of the system (\ref{SystemOper02}) is summarized in the following theorem.
\begin{Theorem} \label{DtNwellposed}
The operator $(I + P) : \prod_{j=1}^{J} H^{s}(\Gamma_{j}) \to \prod_{j=1}^{J} H^{s}(\Gamma_{j})$ is boundedly invertible.
\end{Theorem}

\begin{proof}
Since the obstacles are disjoint then $\Phi(x,y)$ is smooth for all $x \in \Gamma_{i}$ and $y \in \Gamma_{j}$ whenever $i \neq j$. Notice that the first term in (\ref{PropagP}) is an integral operator with kernel $\partial_{n} \Phi(x,y)$. The second term in (\ref{PropagP}) is the composition of $\Lambda_{j}$ and an integral operator with kernel $\Phi(x,y)$. Hence, the propagator operator $P_{ij}$ defined by (\ref{PropagP}) is compact because the single-DtN operator $\Lambda_{j} : H^{s}(\Gamma_{j}) \to H^{s-1}(\Gamma_{j})$ is bounded and the integral operators are highly smoothing since the arguments of $\Phi(x,y)$ belong to disjoint surfaces. This makes the operator matrix $P$ compact and $(I+P)$ bounded. Therefore the Riesz-Fredholm theory \cite{McLean2000,Kress01} applies to the operator $(I+P)$. The well-posedness of the BVP (\ref{BVPsc1})-(\ref{BVPsc3}) and Theorem \ref{TheoremDecomp} imply that the system (\ref{System01}) or equivalent equation (\ref{SystemOper02}) has a solution for all $f \in \prod_{j=1}^{J} H^{s}(\Gamma_{j})$. This implies that $(I+P)$ is surjective, and by the Riesz-Fredholm theory then it is also injective and its inverse is bounded.
\end{proof}

The multiple-OSRC is obtained by replacing the single-DtN operator $\Lambda_{j}$ by an explicit approximation, which we denote by $\tilde{\Lambda}_{j}$. So the precise definition of a multiple-OSRC operator is given as follows.
\begin{Definition}[Multiple-OSRC] \label{Def.MultipleOSRC}
Given a suitable approximation $\tilde{\Lambda}_{j}$ of the single-DtN operator $\Lambda_{j}$ for each $j=1,2,...,J$, the multiple-OSRC operator is defined to be $(I+\tilde{P})$ where $\tilde{P}$ is given by (\ref{PropagP}) and (\ref{BigPropagP}) with $\Lambda_{j}$ replaced by $\tilde{\Lambda}_{j}$. 
\end{Definition}

%%%%%%%%%%%%%%%%%%%%%%%%%%%%%%%%%%%%%%%%%%%%%%%%%%%%%%%%%%%%%%%%%%%%
%%   S E C T I O N
%%%%%%%%%%%%%%%%%%%%%%%%%%%%%%%%%%%%%%%%%%%%%%%%%%%%%%%%%%%%%%%%%%%%

\section{The adjoint Dirichlet-to-Neumann operator} \label{Section.AdjointDtN}

For practical purposes, it is convenient to formulate the multiple-OSRC in terms of the adjoint of the single-DtN map $\Lambda_{j}$. We denote the adjoint operator by $\Lambda_{j}^{*} : H^{1-s}(\Gamma_{j}) \to H^{-s}(\Gamma_{j})$ such that $ \la  \Lambda_{j} v , w \ra_{j} = \la v , \Lambda_{j}^{*} w \ra_{j} = \clo{\la  \Lambda_{j}^{*} w , v \ra}_{j}$ for all $v,w \in H^{s}(\Gamma_{j})$. In that case, the propagation operators, 
$P_{ij} : H^{s}(\Gamma_{j}) \to H^{s}(\Gamma_{i})$ given by (\ref{PropagP}) can be equivalently defined as
\begin{equation}
(P_{ij} v)(x) = \la \partial_{n} \clo{\Phi}(x,\cdot) - \Lambda_{j}^{*} \clo{\Phi}(x,\cdot), v \ra_{j}, \qquad x \in \Gamma_{i}, \qquad i \neq j. 
\label{PropagPAdj} 
\end{equation}

Notice that in (\ref{PropagPAdj}), as opposed to (\ref{PropagP}), no operator is acting on the field $v \in H^{s}(\Gamma_{j})$. This is advantageous in the sense that the field $v$ in (\ref{PropagPAdj}) represents one of the components of the \textit{unknown} solution for the multiple-OSRC system (\ref{SystemOper01}). Instead, we are left to compute $ \Lambda_{j}^{*} \clo{\Phi}(x,\cdot)$, the action of the adjoint single-DtN operator $\Lambda_{j}^{*}$ on the complex-conjugate of the \textit{well-known} fundamental solution $\Phi(x,\cdot)$.  

Now, it remains to characterize the adjoint single-DtN operator in order to obtain useful approximations of it such as the well-known absorbing boundary conditions already discussed in Sections \ref{Section.Intro} and \ref{Section.Formulation}. We accomplish this by explicitly expressing the adjoint DtN operator $\Lambda_{j}^{*}$ in terms of the original DtN operator $\Lambda_{j}$ as follows.
\begin{Theorem} \label{AdjDtN}
Let $\Lambda_{j}^{*} : H^{1-s}(\Gamma_{j}) \to H^{-s}(\Gamma_{j})$ be the adjoint of $\Lambda_{j}$. Then, $\Lambda_{j}^{*} v = \clo{\Lambda_{j} \clo{v}}$ for all $v \in H^{s}(\Gamma_{j})$. 
\end{Theorem}
\begin{proof} 
Let $w,v \in H^{s}(\Gamma_{j})$ be arbitrary. Let $W, V \in H_{\rm loc}^{s+1/2}(\Omega_{j}^{+})$ be the unique generalized solutions of the following problems,
\begin{align*}
& \Delta W + k^2 W = 0 					&& \Delta V + k^2 V = 0 \\
& \gamma_{j}W = w 					&& \gamma_{j}V = v \\
& \partial_{r} W - i k W = \mathcal{O}(1/r^2)		&& \partial_{r} V + i k V = \mathcal{O}(1/r^2)
\end{align*}
Notice from the last two conditions that $W$ is \textit{outgoing} and $V$ is \textit{incoming}. Also notice that $\clo{V}$ is an outgoing solution of the Helmholtz in $\Omega_{j}^{+}$ satisfying the Dirichlet condition $\gamma_{j} \clo{V} = \clo{v}$ on $\Gamma_{j}$. 

An application of Green's second identity \cite{McLean2000,ColtonKress02} to both outgoing fields $W$ and $\clo{V}$ yields,
\begin{eqnarray}
\int_{\Gamma_{j}} \partial_{n} W \,\, \gamma_{j} \clo{V} \,\, dS =  \int_{\Gamma_{j}} \partial_{n} \clo{V} \,\, \gamma_{j}W \,\,dS. \notag
\end{eqnarray}
Now, from the definition of the DtN map, we have that $\Lambda_{j} w = \partial_{n} W$ and $\Lambda_{j} \clo{v} = \partial_{n} \clo{V}$. So in terms of the sesquilinear form, we obtain 
\begin{eqnarray}
\la \Lambda_{j} w , v  \ra_{j} = \la \Lambda_{j} \clo{v} , \clo{w} \ra_{j} = \la w , \clo{\Lambda_{j} \clo{v}} \ra_{j} \quad \text{for all $w,v \in H^{s}(\Gamma_{j})$,} \notag
\end{eqnarray}
which reveals that $\Lambda_{j}^{*} v = \clo{\Lambda_{j} \clo{v}}$ for all $v \in H^{s}(\Gamma_{j})$ as desired. 
\end{proof}

\begin{Remark}
From the proof of Theorem \ref{AdjDtN} we see that the adjoint DtN map is associated with incoming fields in the same manner as the DtN map is associated with the outgoing counterparts. More precisely, using the notation in the proof above, notice that $\Lambda_{j}^{*} v = \clo{\Lambda_{j} \clo{v}} = \clo{\partial_{n} \clo{V}} = \partial_{n} V$. Therefore, it follows that the adjoint operator $\Lambda_{j}^{*}$ maps the Dirichlet data $\gamma_{j} V$ into the Neumann data $ \partial_{n} V$ of the \textit{incoming} wave field $V$.
\end{Remark}

From (\ref{PropagPAdj}) and Theorem \ref{AdjDtN}, we see that the propagation operator $P_{ij}$ can be expressed as a surface integral operator as follows
\begin{equation}
\begin{aligned}
& (P_{ij} v)(x) =  \int_{\Gamma_{j}} K_{ij}(x,y) v(y) \, dS(y), \quad x \in \Gamma_{i} \\ 
& K_{ij}(x,y) = \frac{\partial \Phi(x,y) }{\partial n(y)} - \Lambda_{j} \Phi(x,y), \quad x \in \Gamma_{i} , \quad y \in \Gamma_{j}. 
\end{aligned} \label{PropagInt}
\end{equation}

Here the DtN map $\Lambda_{j}$ associated with a \textit{single-radiation problem} can be explicitly approximated by the standard absorbing boundary conditions \cite{GivoliReview2,Tsynkov1998,GivoliBook,Ihlenburg,Kriegsmann87,Engquist01,Bayliss01,Higdon01}. Recall that most of these absorbing boundary conditions involve tangential and/or normal derivatives. The order of these derivatives usually increases to obtain better performance \cite{Givoli2001}. Hence, it is often the case that the approximate DtN map $\tilde{\Lambda}_{j}$ yields an \textit{unbounded} operator in the normed spaces under consideration.

However, the representation (\ref{PropagInt}) renders another advantage, coming from the fact that the DtN operator $\Lambda_{j}$ acts on the fundamental solution $\Phi$, and not on the unknown function $v$. Recall that $\Phi = \Phi(x,y)$ is smooth (analytic) for arguments $x \in \Gamma_{i}$ and $y \in \Gamma_{j}$ with $i \neq j$. Hence, the approximation $\tilde{\Lambda}_{j} \Phi(x,y)$ simply yields another smooth integral kernel which leaves the boundedness (and actually compactness) of the propagation operator $P_{ij}$ intact. This represents a tremendous advantage from both theoretical and practical points of view.

Notice that the true solution $u$ of the system (\ref{SystemOper02}) is approximated by the solution $\tilde{u}$ of the perturbed system according to the Def. \ref{Def.MultipleOSRC} of the multiple-OSRC method. To ensure that $\tilde{u}$ exists and that $\| u - \tilde{u} \|$ is sufficiently small, we may resort to either one of the following options:
\begin{itemize}
 \item[(a)] Enforce that $\| \Lambda_{j} - \tilde{\Lambda}_{j} \|$ is sufficiently small in the \textit{operator} norm.
 \item[(b)] Enforce that $\| \Lambda_{j} \Phi - \tilde{\Lambda}_{j} \Phi \|$ is sufficiently small in some appropriate norm.  
\end{itemize}
Notice that option (a) cannot be enforced in general since $\tilde{\Lambda}_{j}$ may be an unbounded operator as explained above. And even if the approximate DtN operator $\tilde{\Lambda}_{j}$ is chosen to map boundedly into the appropriate space, it is much easier to ensure the smallness of a vector-norm, such as in option (b), than the smallness of an operator-norm, such as option (a) above. Again, this second option is possible because we have proposed to apply the DtN operator to a \textit{fixed} function $\Phi$ as opposed to the unknown function $v$. In other words, the operator $\tilde{\Lambda}_{j}$ must be chosen to act satisfactorily only on the fundamental solution $\Phi$ rather than the whole space $H^{s}(\Gamma_{j})$ which may possible contain highly oscillatory or non-smooth functions.

%%%%%%%%%%%%%%%%%%%%%%%%%%%%%%%%%%%%%%%%%%%%%%%%%%%%%%%%%%%%%%%%%%%%
%%   S E C T I O N
%%%%%%%%%%%%%%%%%%%%%%%%%%%%%%%%%%%%%%%%%%%%%%%%%%%%%%%%%%%%%%%%%%%%

\section{Orders of Scattering} \label{Section.Iterative}

When the obstacles are sufficiently small or far away from each other (with respect to the wavelength), we may show that the propagation operator $P$, defined in (\ref{PropagInt}), has a sufficiently small norm so that (\ref{SystemOper02}) can be solved using the Neumann series. This is easily seen from the following estimates,
\begin{eqnarray*}
&& \frac{\partial \Phi(x,y)}{\partial n(y)} = \frac{e^{i k |x-y|}}{4 \pi |x-y|} \left( \frac{ (x-y) \cdot n(y) }{ |x-y|^2} -  ik \frac{ (x-y) \cdot n(y) }{ |x-y| }  \right), \qquad x \in \Gamma_{i}, \quad y \in \Gamma_{j} \\
&& | \Lambda_{j} \Phi(x,\cdot) | \leq \frac{C}{\text{dist}(x,\Gamma_{j})} 
\end{eqnarray*}
where the constant $C=C(k)$ is independent of $x$ and grows with $k$. For disjoint obstacles, a simple calculations shows that
\begin{eqnarray}
\| P_{ij} \| \leq \, C(k) \, \frac{|\Gamma_{i}|^{1/2} |\Gamma_{j}|^{1/2} }{\text{dist}(\Gamma_{i},\Gamma_{j})}, \label{WeakScatt}
\end{eqnarray}
for some other constant $C=C(k)$ which grows as $\sim k$. For a fixed wavenumber $k$, if the distance between $\Gamma_{i}$ and $\Gamma_{j}$ is sufficiently large, or the surface measures of $\Gamma_{i}$ and $\Gamma_{j}$ are sufficiently small, then the norm of $P_{ij}$ will be less than $1$. If the smallness in the norm of $P_{ij}$ is uniform for all $i,j = 1,...,J$ then we have that $\| P \| < 1$ as well. In that case, we have that the solution of (\ref{SystemOper02}) is given by
\begin{eqnarray*}
u = \sum_{n=0}^{\infty} (-1)^n P^n f
\end{eqnarray*}
with convergence in the norm of $\prod_{j=1}^{J} H^{0}(\Gamma_{j})$ defined in (\ref{Norm}). Unfortunately, the constant $C$ appearing in (\ref{WeakScatt}) is non-trivially dependent on the geometry of the obstacles which makes it difficult to estimate a-priori. This renders inequality (\ref{WeakScatt}) hard to use in practice as a precise way to verify the smallness of the propagation operator $P$.

We refer to Chapter 8 in \cite{MartinBook} for a review of the orders of scattering theory widely employed in multiple scattering problems. Recall that in practice $P$ is replaced by its approximation $\tilde{P}$ using approximations of the single-DtN maps $\Lambda_{j}$. Thus, we obtain an approximate solution 
\begin{eqnarray}
\tilde{u} = \sum_{n=0}^{N} (-1)^n \tilde{P}^n f \label{NeumannApprox}
\end{eqnarray}
where the series has been truncated. The error $\| u - \tilde{u} \|$ can only be controlled by increasing $N$ and using better approximations of the propagator $P$. As we mentioned in the Introduction, the application of (\ref{NeumannApprox}) does not require to solve single-scattering problems which is inherited from the original OSRC. In practice, one may equivalently arrive at (\ref{NeumannApprox}) from the following iterative scheme, $\tilde{u}_{n+1} = - \tilde{P} \tilde{u}_{n}$ and $\tilde{u}_{0} = f$. The application of this iteration only involves the (numerical) integration associated with the applications of $\tilde{P}$ with no need to build or invert a matrix. As a result, this method renders an approximate solution at a low computational cost, assuming that the weak-scattering condition is well-satisfied. 

We finish this Section with some comments concerning the validity of the \textit{orders of scattering} approach in two-dimensional scenarios. We realize that the estimate (\ref{WeakScatt}) is only valid in three dimensions because the fundamental solution $\Phi(x,y)$ decays as $|x-y|^{-1}$. In two dimensions, the fundamental solution decays as $|x-y|^{-1/2}$. Hence, we should expect the wide-spacing approximation to work better in three-dimensional problems than in two dimensions.

%%%%%%%%%%%%%%%%%%%%%%%%%%%%%%%%%%%%%%%%%%%%%%%%%%%%%%%%%%%%%%%%%%%%
%%   S E C T I O N
%%%%%%%%%%%%%%%%%%%%%%%%%%%%%%%%%%%%%%%%%%%%%%%%%%%%%%%%%%%%%%%%%%%%

\section{The far-field pattern} \label{Section.FFP}

In this section, we explicitly review the definition of the so-called far-field pattern corresponding to the scattered field. This is employed in the next section to compare exact and numerical solutions. It is well-known that the scattered field $\usc$ admits the following asymptotic behavior
\begin{eqnarray*}
\usc(x) = \frac{e^{i k |x|}}{4 \pi |x|} \left( u^{\infty}(\hat{x}) + \mathcal{O}(|x|^{-1})  \right), \qquad \hat{x} = x / |x|
\end{eqnarray*}
where $u^{\infty}$ is known as the far-field pattern of $\usc$. An analogous asymptotic behavior holds for each purely outgoing wave field $u_{j}$ defined in Theorem \ref{TheoremDecomp}. From the asymptotics of the fundamental solution
\begin{eqnarray*}
\Phi(x,y) = \frac{e^{i k |x|}}{4 \pi |x|} \left(  e^{- i k \hat{x} \cdot y} + \mathcal{O}(|x|^{-1}) \right) \quad \text{and} \quad
\frac{\partial \Phi(x,y)}{\partial n(y)} = \frac{ e^{i k |x|} }{4 \pi |x|} \left( - ik \hat{x} \cdot n(y) e^{- i k \hat{x} \cdot y} + \mathcal{O}(|x|^{-1}) \right)
\end{eqnarray*}
and the representation (\ref{ExplicitV01}), we obtain the far-field pattern for each purely-outgoing wave field to be
\begin{eqnarray}
u^{\infty}_{j}(\hat{x}) = \int_{\Gamma_{j}} \bigg[ - i k \hat{x} \cdot n(y) e^{ - ik \hat{x} \cdot y} - \Lambda_{j} e^{ - ik \hat{x} \cdot y} \bigg] u_{j}(y) dS(y).  \label{ApproxFFP}
\end{eqnarray}
Therefore, once each purely-outgoing wave field $u_{j}$ is approximated using the multiple-OSRC, then we obtain an approximation to the multiple-scattering far-field pattern given by
\begin{eqnarray}
u^{\infty} = \sum_{j=1}^{J} u^{\infty}_{j}.  \label{MultipleFFP}
\end{eqnarray}

%%%%%%%%%%%%%%%%%%%%%%%%%%%%%%%%%%%%%%%%%%%%%%%%%%%%%%%%%%%%%%%%%%%%
%%   S E C T I O N
%%%%%%%%%%%%%%%%%%%%%%%%%%%%%%%%%%%%%%%%%%%%%%%%%%%%%%%%%%%%%%%%%%%%

\section{Numerical examples}

In this section we present the numerical results obtained from the implementation of the proposed multiple-OSRC. We discuss the following two examples. 

\subsection{Example 1} \label{SubSection.Ex1}
Here we only consider two obstacles embedded in the three-dimensional space which are shaped and located axisymmetrically with respect to the $z$-axis. We define an incident field as to easily obtain the exact solution for this problem. More precisely, the incident field is given by
\begin{eqnarray}
\uinc(x) = \Phi(x,c_{1}) + \Phi(x,c_{2}).  \label{Eqn.Uinc}
\end{eqnarray}
This is the superposition of two point sources with respective centers at $c_{1}$ and $c_{2}$. Each center point $c_{j}$ is purposely located within the obstacle $\Omega_{j}^{-}$. Notice that each point source $\Phi(\cdot,c_{j})$ is a radiating solution of the Helmholtz equation in $\Omega_{j}^{+}$. As a consequence, a simple verification shows that each purely-outgoing scattered field is exactly given by $u_{j}(x) = - \Phi(x,c_{j})$ in $\Omega_{j}^{+}$. This latter statement is true regardless of the shape of the obstacle $\Omega_{j}^{-}$ as long as $c_{j} \in \Omega_{j}^{-}$. We also obtain the exact far-field pattern to be $u^{\infty}(\hat{x}) = - (e^{- i k \hat{x} \cdot c_{1}} + e^{- i k \hat{x} \cdot c_{2}})$. For the numerical results shown below we have chosen $c_{1} = (0,0,2)$ and $c_{2} = (0,0,-2)$ expressed in Cartesian coordinates.

Now, we consider explicit expressions to approximate the single-DtN map $\Lambda_{j}$ acting on the fundamental solution $\Phi(x,y)$. For instance, from \cite{Alber-1981} we can extract the following approximation valid for high-frequencies. In principle the OSRC is not limited to high-frequency waves, but most well-known local approximations of the single-DtN map perform better as the frequency increases due to the asymptotic localization of the solution in the high-frequency regime. See details in \cite{Antoine99}. The numerical results of this section are obtained using Theorem \ref{Thm.Alber} to approximate each single-DtN operator $\Lambda_{j}$ by neglecting the higher-order terms.

\begin{Theorem} \label{Thm.Alber}
Let $x \in \Gamma_{i}$ and $y \in \Gamma_{j}$.
\begin{itemize}
\item[(i)] If none of the rays starting at $x$ passes over $y$, ie. $y$ belongs to the portion of $\Gamma_{j}$ not illuminated by a point source at $x$, then for any $m > 1$
\begin{eqnarray*}
\Lambda_{j} \Phi(x,y) = \Phi(x,y) \left( \frac{ (x-y) \cdot n(y) }{ |x-y|^2} -  ik \frac{ (x-y) \cdot n(y) }{ |x-y| }  \right) + \mathcal{O}(k^{-m}).
\end{eqnarray*}

\item[(ii)] If there is a ray passing over $y$ which is not tangent to $\Gamma_{j}$ at $y$, ie. $y$ belongs to the portion of $\Gamma_{j}$ illuminated by a point source at $x$, then
\begin{eqnarray*}
\Lambda_{j} \Phi(x,y) = \Phi(x,y) \left( ik \frac{(x-y)\cdot n(y)}{|x-y|} + \frac{(x-y)\cdot n(y)}{|x-y|^2} -  \frac{1}{2 (x-y) \cdot n(y)  } \right) + \mathcal{O}(k^{-1}).  
\end{eqnarray*}
\end{itemize}
\end{Theorem}

The boundary of $\Omega_{1}^{-}$ is defined as a surface of revolution from the following parametric curve described in axisymmetric cylindrical $(r,z)$ coordinates,
\begin{eqnarray}
r(t) = \sin t, \qquad z(t) = 2 + \cos t + \frac{1}{2} \cos 2t, \qquad t \in [0,\pi], \label{Eqn.Obs}
\end{eqnarray}
and the boundary of $\Omega_{2}^{-}$ is a mirror imaging of  $\Omega_{1}^{-}$ about the $z=0$ plane. An illustration of these obstacles is displayed in Figure \ref{Fig:Obstacles}. Notice that they are not entirely convex. Yet the multiple-OSRC renders good results as displayed in Figure \ref{Fig:FFP-6pi} and Table \ref{Table:Error}.

\begin{figure}[ht]
\centering
\includegraphics[height=0.5\textwidth]{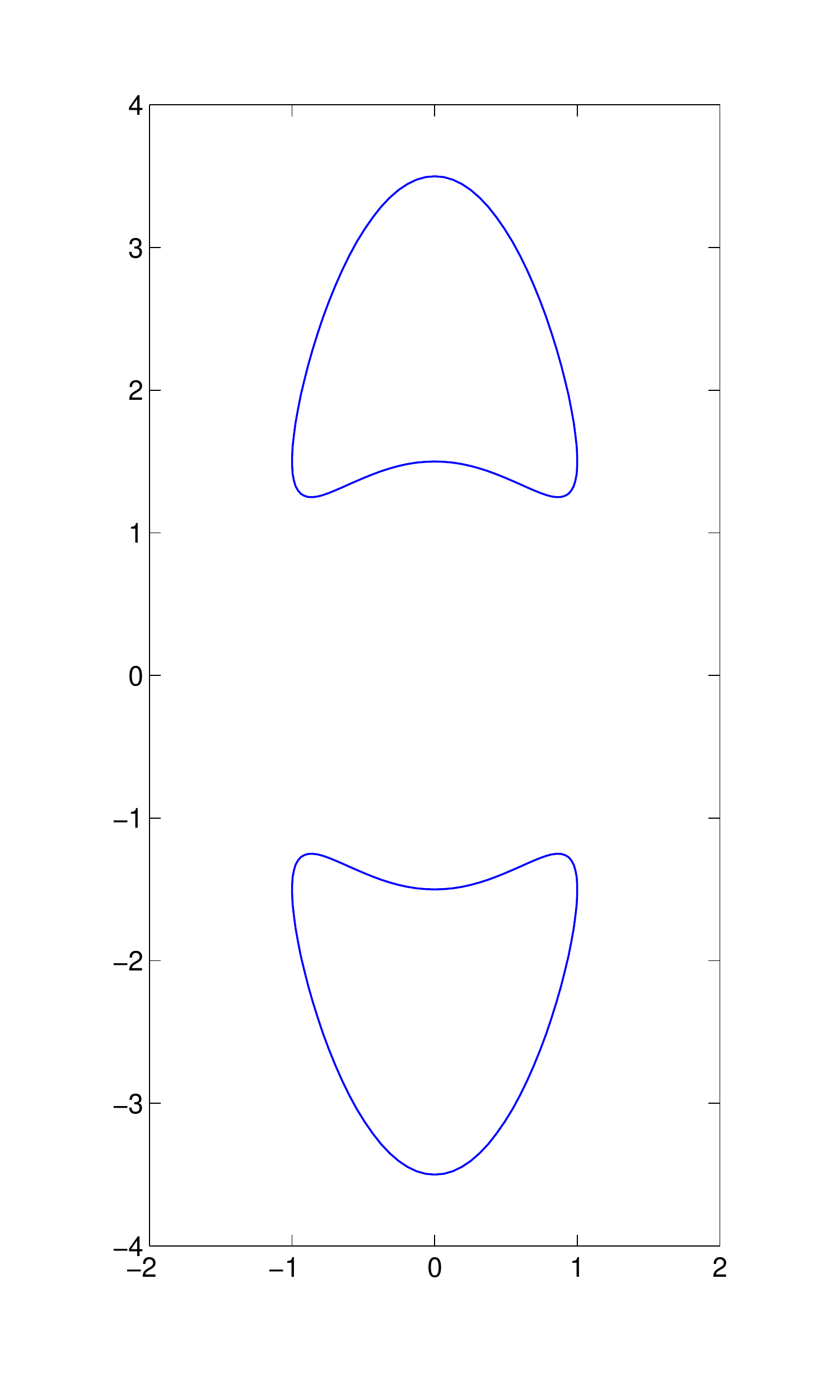} 
\caption{Illustration of the axisymmetric obstacles' cross section defined in (\ref{Eqn.Obs}).}
\label{Fig:Obstacles}
\end{figure}

For the incident field (\ref{Eqn.Uinc}), we computed the far-field pattern using the proposed multiple-OSRC. A comparison with the exact solution is displayed in Figure \ref{Fig:FFP-6pi} for wavenumber $k = 8 \pi$. Results for various values of $k$ are displayed in Table \ref{Table:Error} where the following relative error in the $L^2$-norm is reported,
\begin{eqnarray}
E(k) = \frac{\| \tilde{u}^{\infty} - u^{\infty} \|_{L^2}}{\| u^{\infty} \|_{L^2}}. \label{Eqn.RelError}
\end{eqnarray}
The meshes employed to discretize the axisymmetric surfaces were chosen to contain at least $8$ points per wave length in order to properly resolve the oscillatory behavior of the waves fields. We also note that Table \ref{Table:Error} reflects the fact that the error decreases as the frequency increases due to the asymptotic behavior ensured by Theorem \ref{Thm.Alber}. It is also worth mentioning here that the \textit{orders of scattering} iterative procedure described in Section \ref{Section.Iterative} was employed to obtain the numerical solutions for the scattering configurations of this section.

\begin{figure}[ht]
\centering
\includegraphics[height=0.5\textwidth]{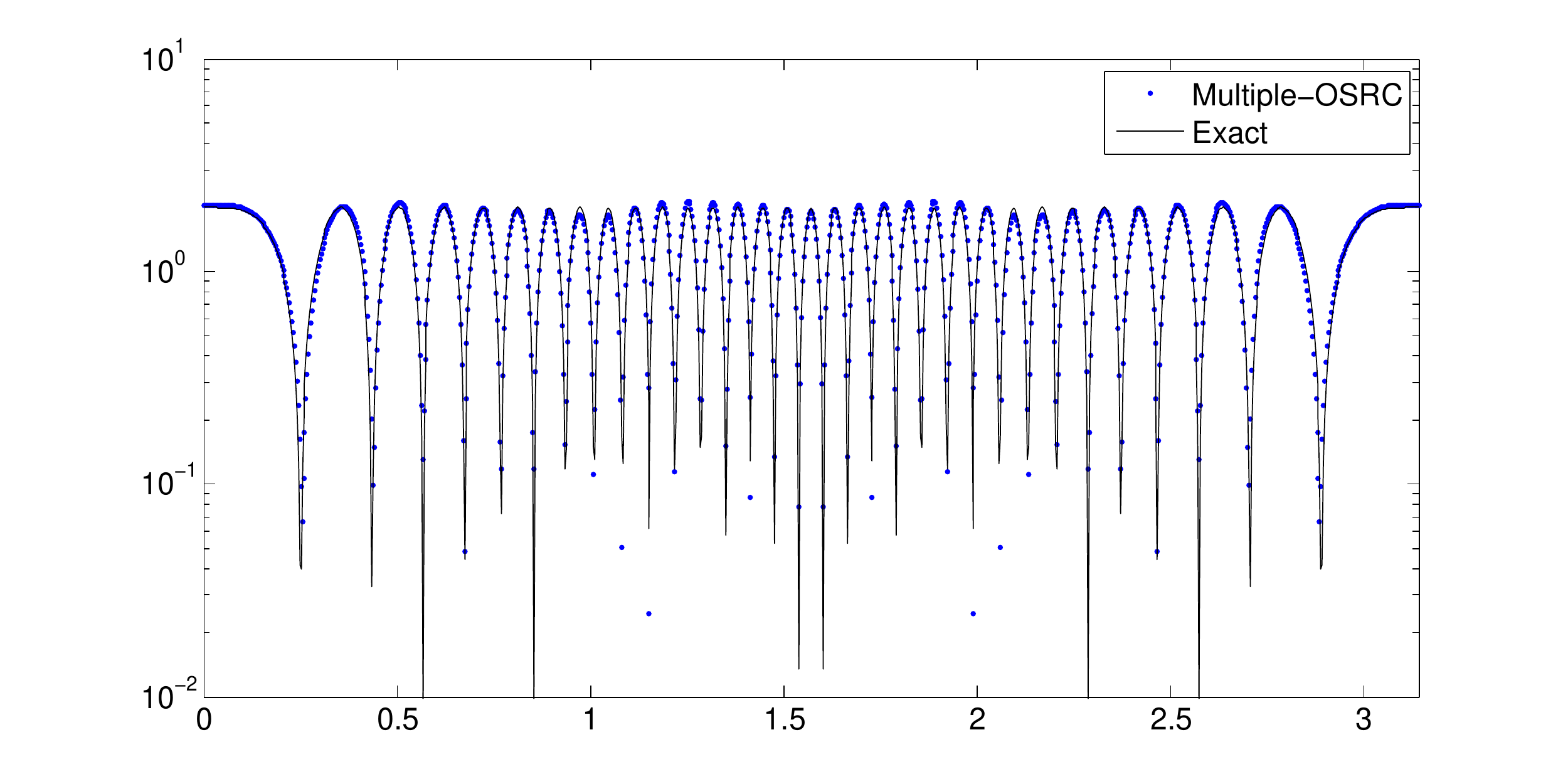} 
\caption{Comparison of absolute values of the far-field patterns for $k=8 \pi$. The horizontal axis represents the zenith angle measured from the $z$-axis for the axisymmetric problem presented in Subsection \ref{SubSection.Ex1}.}
\label{Fig:FFP-6pi}
\end{figure}

\begin{table}[ht]
\caption{Relative error (\ref{Eqn.RelError}) for various values of wavenumber $k$ for the problem presented in Subsection \ref{SubSection.Ex1}.}
\centering
\begin{tabular}{ c c c c c }
\hline \hline
$k=4\pi$ & $k=8\pi$ & $k=12\pi$ & $k=16\pi$ & $k=20\pi$ \\ 
\hline
1.74e-1 & 1.06e-1 & 7.45e-2 & 5.83e-2 & 4.99e-2  \\ 
\hline 
\end{tabular} 
\label{Table:Error}
\end{table}

\subsection{Example 2} \label{SubSection.Ex2}
Here we consider four obstacles embedded in two-dimensional space. They are illustrated in Figure \ref{Fig:Obstacles4}. As in the previous example, we define the incident field as to easily obtain the exact solution of this problem. The incident field is given by
\begin{eqnarray}
\uinc(x) = \Phi(x,c_{1}) + \Phi(x,c_{2}) + \Phi(x,c_{3}) + \Phi(x,c_{4}),  \label{Eqn.Uinc2}
\end{eqnarray}
where the two-dimensional fundamental solution is $\Phi(x,y) = -i/4 H^{(1)}_{0}(k|x-y|)$, and $H^{(1)}_{0}$ is the zeroth order Hankel function of the first kind. The incident field (\ref{Eqn.Uinc2}) is the superposition of four point sources with centers at $c_{1} = (2,2)$, $c_{2} = (2,-2)$, $c_{3} = (-2,-2)$ and $c_{4} = (-2,2)$, expressed in Cartesian coordinates. The exact far-field pattern is easily obtained to be $u^{\infty}(\hat{x}) = - \sum_{j=1}^{4} e^{- i k \hat{x} \cdot c_{j}}$. 

As before, we consider explicit expressions to approximate the single-DtN map $\Lambda_{j}$. For this example, we consider a couple of Bayliss-Turkel-like absorbing boundary conditions derived in \cite{Antoine99}. These are given by
\begin{eqnarray} 
& \text{1st order} \qquad & \tilde{\Lambda} u = ik u - \frac{\mathscr{C}}{2} u \label{Eqn.order1} \\
& \text{2nd order} \qquad & \tilde{\Lambda} u = ik u - \frac{\mathscr{C}}{2} u + \frac{i \mathscr{C}^2}{8 k (1 + i \mathscr{C} / k )} u + \frac{\partial_{s}^2 \mathscr{C}}{8 k^2} u - \partial_{s} \left( \frac{1}{2 i k (1 + i \mathscr{C} / k )} \partial_{s} \right) u \label{Eqn.order2}
\end{eqnarray}
where $\mathscr{C}$ is the curvature of the obstacle's boundary and $s$ represents its curvilinear abscissa.

For the incident field (\ref{Eqn.Uinc2}), we computed the approximate scattered field and the corresponding far-field pattern using the proposed multiple-OSRC. A comparison with the exact solution is displayed in Figure \ref{Fig:FFP-4pi-4Obs} for wavenumber $k = 4 \pi$. Results in terms of the relative error (\ref{Eqn.RelError}) for various values of $k$ are displayed in Table \ref{Table:Error2}. As expected, we note from this table that the second order approximation of the single-DtN map performs better than the first order. Similar to the example in the previous subsection, the relative error decreases as the frequency increases.

\begin{table}[ht]
\caption{Relative error (\ref{Eqn.RelError}) for various values of wavenumber $k$ and meshes with approx. $8$ points per wavelength.}
\centering
\begin{tabular}{ c c c c c c }
\hline \hline
Approx DtN map & $k=4\pi$ & $k=8\pi$ & $k=12\pi$ & $k=16\pi$ & $k=20\pi$ \\ 
\hline
\text{Eqn. (\ref{Eqn.order1}) } & 2.75e-2 & 1.38e-2 & 9.28e-3 & 7.05e-3 & 5.67e-3  \\ 
\text{Eqn. (\ref{Eqn.order2}) } & 2.70e-2 & 7.81e-3 & 3.70e-3 & 2.19e-3 & 1.47e-3  \\ 
\hline 
\end{tabular} 
\label{Table:Error2}
\end{table}

\begin{figure}[ht]
\centering
\includegraphics[width=0.5\textwidth]{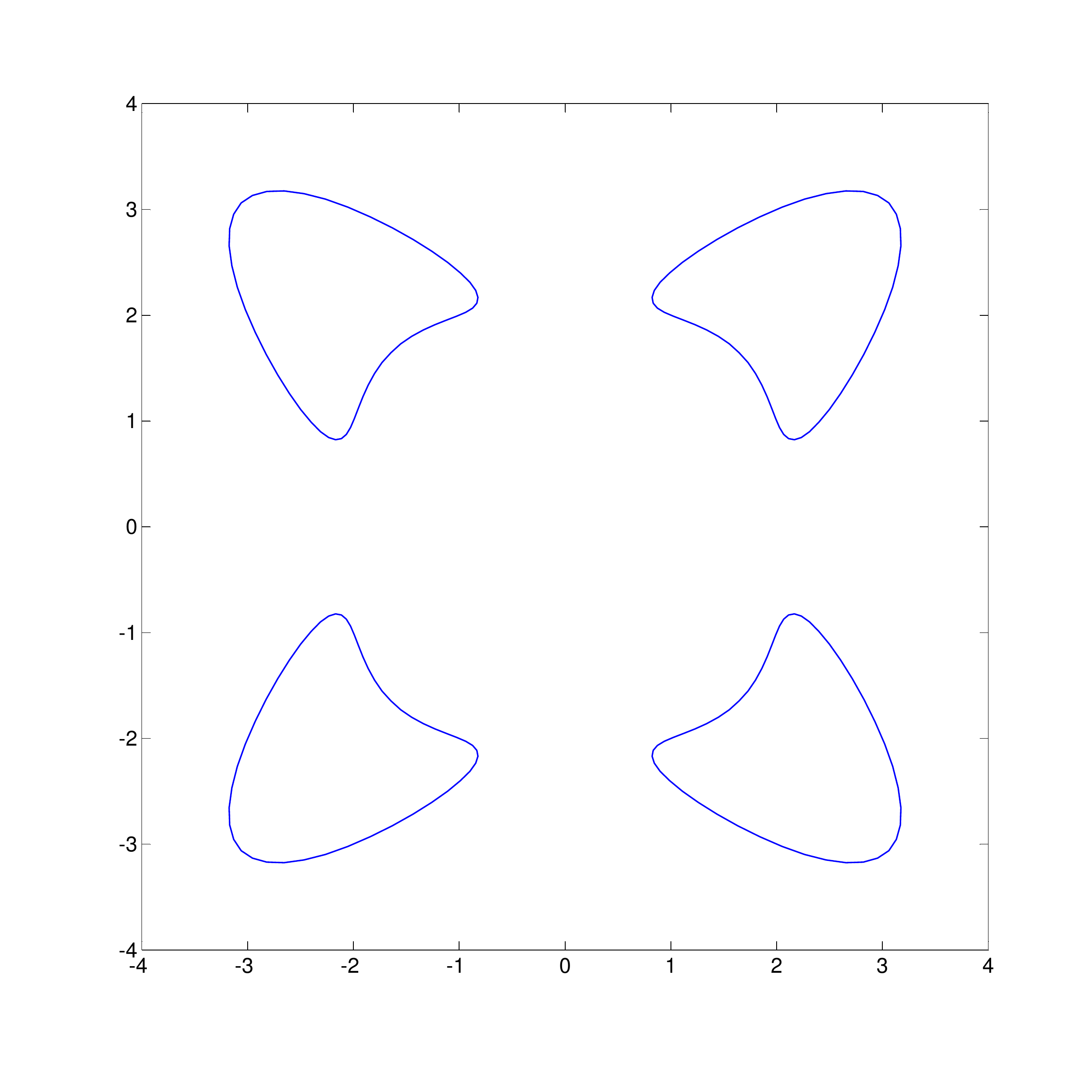} 
\caption{Illustration of the two-dimensional obstacles discussed in Subsection \ref{SubSection.Ex2}.}
\label{Fig:Obstacles4}
\end{figure}

\begin{figure}[ht]
\centering
\includegraphics[width=0.9\textwidth]{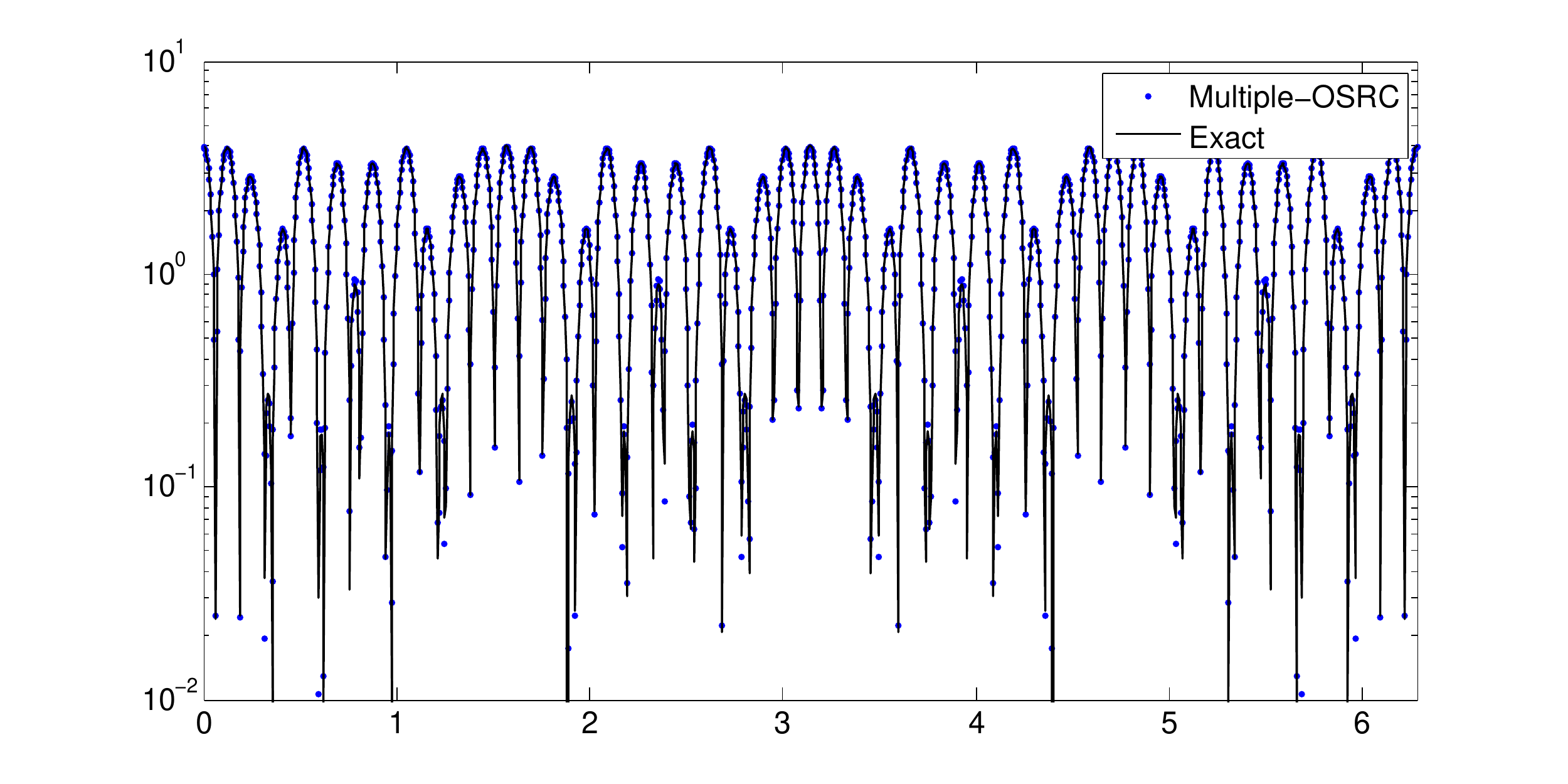} 
\caption{Comparison of absolute values of the far-field patterns for $k=4 \pi$. The horizontal axis represents the polar angle for the two-dimensional problem presented in Subsection \ref{SubSection.Ex2}.}
\label{Fig:FFP-4pi-4Obs}
\end{figure}

%%%%%%%%%%%%%%%%%%%%%%%%%%%%%%%%%%%%%%%%%%%%%%%%%%%%%%%%%%%%%%%%
%%%   S E C T I O N
%%%%%%%%%%%%%%%%%%%%%%%%%%%%%%%%%%%%%%%%%%%%%%%%%%%%%%%%%%%%%%%%
 
\section{Summary and future work}

A novel on-surface radiation condition for scattering from multiple obstacles has been proposed in this work. This condition allows the waves to interact between the several obstacles and propagate towards infinity in the correct physical manner. The multiple-OSRC renders a system of integral equations for the boundary values of the purely-outgoing waves which compose the sought scattered field.
The resulting system of boundary integral equations is Fredholm of the second kind with \textit{singularity-free kernels}. Under weak scattering conditions, the multiple-OSRC leads to a convergent method of successive approximations. At each iteration there is no need to solve a single-scattering problem for each obstacle because the single-OSRC renders an explicit approximate formula (up to numerical integration). 

Since the multiple-OSRC is only a low order approximate method, we highlight 
the potential of this proposed approach to yield extraordinarily good initial guesses and inexpensive pre-conditioners for Krylov iterative solutions of BIE. Similarly, the multiple-OSRC may serve as a fast approximate method to explore parameter-spaces employed in optimization algorithms and reduced-order models where most of the computational effort is spent in the construction of reduced bases \cite{Gan-Hes-Sta-2012}. 

The extension to other boundary conditions on the surface of the obstacles is relatively straightforward. Here again, explicit approximations of the Neumann-to-Dirichlet map are found, for instance in \cite{Engquist02}, to handle problems with a Neumann condition. One may similarly construct useful approximations for the Robin-to-Dirichlet and Robin-to-Neumann maps, and their inverses. Although the following are beyond the scope of this paper, we briefly delineate some possible extensions to enhance the impact of our work on engineering problems. 

\begin{enumerate}
\item{The extension to electrodynamics. The main ingredients in the formulation of the multiple-OSRC are the Decomposition Theorem \ref{TheoremDecomp} and physically meaningful approximations of the single-DtN maps. Fortunately, both of these ingredients are found in the literature for Maxwell's equations. See for instance \cite{Balabane04} and \cite{Engquist02,AntBar01,Amm-He-1998,Barucq-2002}, respectively.}
 
\item{The extension to the time domain. Here again, a decomposition theorem and applicable approximations of the hyperbolic-DtN maps are needed. Approximations for the hyperbolic-DtN map can be obtained by simply recalling the Fourier duality between $\partial_{t}$ and $-ik$. In the time domain, we would need to numerically integrate transient boundary layer potentials. See \cite{Dom-Sayas-2013} for a review and recent analysis of time domain boundary integral operators.}
\end{enumerate}

%%%%%%%%%%%%%%%%%%%%%%%%%%%%%%%%%%%%%%%%%%%%%%%%%%%%%%%%%%%%%%%%%%%%
%   A C K N O W L E D G M E N T S
%%%%%%%%%%%%%%%%%%%%%%%%%%%%%%%%%%%%%%%%%%%%%%%%%%%%%%%%%%%%%%%%%%%%

\section*{Acknowledgments}
The author would like to thank the anonymous referees for their
most constructive suggestions which certainly improved the quality
of the manuscript. The author also acknowledges the support provided by the graduate fellowship at Rice University.

%%%%%%%%%%%%%%%%%%%%%%%%%%%%%%%%%%%%%%%%%%%%%%%%%%%%%%%%%%%%%%%%%%%%
%   B I B L I O G R A P H Y
%%%%%%%%%%%%%%%%%%%%%%%%%%%%%%%%%%%%%%%%%%%%%%%%%%%%%%%%%%%%%%%%%%%%

\bibliographystyle{elsarticle-num}
\bibliography{AcoBib}

\begin{thebibliography}{10}
\expandafter\ifx\csname url\endcsname\relax
  \def\url#1{\texttt{#1}}\fi
\expandafter\ifx\csname urlprefix\endcsname\relax\def\urlprefix{URL }\fi
\expandafter\ifx\csname href\endcsname\relax
  \def\href#1#2{#2} \def\path#1{#1}\fi

\bibitem{MartinBook}
P.~Martin, Multiple Scattering, Cambridge Univ. Press, 2006.

\bibitem{Antoine2008}
X.~Antoine, C.~Chniti, K.~Ramdani, On the numerical approximation of
  high-frequency acoustic multiple scattering problems by circular cylinders,
  J. Comput. Phys. 227 (2008) 1754--1771.

\bibitem{Gau-Hua-Str-1995}
G.~C. Gaunaurd, H.~Huang, H.~C. Strifors, Acoustic scattering by a pair of
  spheres, The Journal of the Acoustical Society of America 98~(1) (1995)
  495--507.

\bibitem{Lee-2012}
W.-M. Lee, Acoustic scattering by multiple elliptical cylinders using
  collocation multipole method, Journal of Computational Physics 231~(14)
  (2012) 4597 -- 4612.

\bibitem{GivoliReview2}
D.~Givoli, High-order non-reflecting boundary conditions : A review, Wave
  Motion 39 (2004) 319--326.

\bibitem{Tsynkov1998}
S.~Tsynkov, Numerical solution of problems on unbounded domains, Appl. Numer.
  Math. 27 (1998) 465--532.

\bibitem{BHPR2010}
A.~Bermudez, L.~Hervella-Nieto, A.~Prieto, R.~Rodriguez, Perfectly matched
  layers for time-harmonic second order elliptic problems, Arch. Comput.
  Methods Engrg. 17 (2010) 77--107.

\bibitem{GivoliBook}
D.~Givoli, Numerical Methods for Problems in Infinite Domains, Vol.~33 of
  Studies in Applied Mechanics, Elsevier, 1992.

\bibitem{Ihlenburg}
F.~Ihlenburg, Finite Element Analysis of Acoustic Scattering, Springer, 1998.

\bibitem{Kriegsmann87}
G.~Kriegsmann, A.~Taflove, K.~Umashankar, A new formulation of electromagnetic
  scattering using an on surface radiation condition approach, IEEE Trans. Ant.
  Prop. AP35 (1987) 153--161.

\bibitem{Givoli2001}
D.~Givoli, High-order nonreflecting boundary conditions without high-order
  derivatives, J. Comput. Phys. 170 (2001) 849--870.

\bibitem{Rab-Giv-Bec-2010}
D.~Rabinovich, D.~Givoli, E.~B\'{e}cache, Comparison of high-order absorbing
  boundary conditions and perfectly matched layers in the frequency domain,
  Int. J. Numerical Methods in Biomedical Engineering 26~(10) (2010)
  1351--1369.

\bibitem{Grote01}
M.~Grote, C.~Kirsch, \textsc{D}irichlet-to-\textsc{N}eumann boundary conditions
  for multiple scattering problems, J. Comput. Phys. 201 (2004) 630--650.

\bibitem{Grote02}
M.~Grote, C.~Kirsch, Nonreflecting boundary condition for time-dependent
  multiple scattering, J. Comput. Phys. 221 (2007) 41--62.

\bibitem{Grote03}
M.~Grote, I.~Sim, Local nonreflecting boundary condition for time-dependent
  multiple scattering, J. Comput. Phys. 230 (2011) 3135--3154.

\bibitem{JiangZheng}
X.~Jiang, W.~Zheng, Adaptive perfectly matched layer method for multiple
  scattering problems, Comput. Methods Appl. Mech. Engrg. 201-204 (2012)
  42--52.

\bibitem{Mart-Rok-2007}
P.~Martinsson, V.~Rokhlin, A fast direct solver for scattering problems
  involving elongated structures, Journal of Computational Physics 221~(1)
  (2007) 288--302.

\bibitem{Ant-Dar-2005}
X.~Antoine, M.~Darbas, Alternative integral equations for the iterative
  solution of acoustic scattering problems, The Quarterly Journal of Mechanics
  and Applied Mathematics 58~(1) (2005) 107--128.

\bibitem{Ant-Dar-2007}
X.~Antoine, M.~Darbas, Generalized combined field integral equations for the
  iterative solution of the three-dimensional \textsc{H}elmholtz equation,
  Mathematical Modelling and Numerical Analysis 41 (2007) 147--167.

\bibitem{Gan-Hes-Sta-2012}
M.~Ganesh, J.~Hesthaven, B.~Stamm, A reduced basis method for electromagnetic
  scattering by multiple particles in three dimensions, Journal of
  Computational Physics 231~(23) (2012) 7756 -- 7779.

\bibitem{Ant-InBook2008}
X.~Antoine, Advances in the on-surface radiation condition method: Theory,
  numerics and applications, in: F.~Magoul\`{e}s (Ed.), Comput. Meth. for
  Acoustics Problems, Saxe-Coburg Publ. Stirlingshire, UK, 2008, pp. 207--232.

\bibitem{Antoine01}
X.~Antoine, Fast approximate computation of a time-harmonic scattered field
  using the on-surface radiation condition method, IMA J. Appl. Math. 66 (2001)
  83--110.

\bibitem{Antoine06}
X.~Antoine, M.~Darbas, Y.~Lu, An improved surface radiation condition for
  high-frequency acoustic scattering problems, Comput. Methods Appl. Mech.
  Engrg. 195 (2006) 4060--4074.

\bibitem{Engquist02}
A.~Atle, B.~Engquist, On surface radiation conditions for high-frequency wave
  scattering, J. Comp. Appl. Math. 204 (2007) 306--316.

\bibitem{Antoine99}
X.~Antoine, H.~Barucq, A.~Bendali, \textsc{B}ayliss-\textsc{T}urkel like
  radiation conditions on surfaces of arbitrary shape, J. Math. Anal. Appl. 229
  (1999) 184--211.

\bibitem{AntBar01}
X.~Antoine, H.~Barucq, Microlocal diagonalization of strictly hyperbolic
  pseudodifferential systems and application to the design of radiation
  conditions in electromagnetism, SIAM J. Appl. Math. 61 (2001) 1877--1905.

\bibitem{BarDiazDup12}
H.~Barucq, J.~Diaz, V.~Duprat, Micro-differential boundary conditions modelling
  the absorption of acoustic waves by 2\textsc{D} arbitrarily-shaped convex
  surfaces, Commun. Comput. Phys. 11 (2012) 674--690.

\bibitem{Alber-1981}
H.-D. Alber, Justification of geometrical optics for non-convex obstacles, J.
  Math. Anal. Appl. 80 (1981) 372--386.

\bibitem{ColtonKress02}
D.~Colton, R.~Kress, Inverse Acoustic and Electromagnetic Scattering Theory,
  2nd Edition, Springer, 1998.

\bibitem{Nedelec01}
J.~Nedelec, Acoustic and Electromagnetic Equations : Integral Representations
  for Harmonic Problems, Springer, 2001.

\bibitem{McLean2000}
W.~McLean, Strongly Elliptic Systems and Boundary Integral Equations, Cambridge
  Univ. Press, 2000.

\bibitem{Engquist01}
B.~Engquist, A.~Majda, Absorbing boundary conditions for the numerical
  simulation of waves, Math. Comput. 31 (1977) 629--651.

\bibitem{Bayliss01}
A.~Bayliss, M.~Gunzburger, E.~Turkel, Boundary conditions for the numerical
  solution of elliptic equations in exterior regions, SIAM J. Appl. Math. 42
  (1982) 430--451.

\bibitem{Higdon01}
R.~Higdon, Absorbing boundary conditions for difference approximations to the
  multi-dimensional wave equation, Math. Comput. 47 (1986) 437--459.

\bibitem{Balabane04}
M.~Balabane, Boundary decomposition for the \textsc{H}elmholtz and
  \textsc{M}axwell equations 1: disjoint sub-scatterers, Asymptotic Analysis 38
  (2004) 1--10.

\bibitem{Wang-Liu-2013}
H.~Wang, J.~Liu, On decomposition method for acoustic wave scattering by
  multiple obstacles, Acta Mathematica Scientia 33~(1) (2013) 1 -- 22.

\bibitem{JCP2010}
S.~Acosta, V.~Villamizar, Coupling of \textsc{D}irichlet-to-\textsc{N}eumann
  boundary condition and finite difference methods in curvilinear coordinates
  for multiple scattering, J. Comput. Phys. 229 (2010) 5498--5517.

\bibitem{CoaJoly2012}
J.~Coatl\'{e}ven, P.~Joly, Operator factorization for multiple-scattering
  problems and an application to periodic media, Commun. Comput. Phys. 11
  (2012) 303--318.

\bibitem{Aco-Vill-Mal}
S.~Acosta, V.~Villamizar, B.~Malone, The \textsc{D}t\textsc{N} nonreflecting
  boundary condition for multiple scattering problems in the half-plane,
  Comput. Methods Appl. Mech. Engrg. 217-220 (2012) 1--11.

\bibitem{Kress01}
R.~Kress, Linear Integral Equations, 2nd Edition, Springer, 1999.

\bibitem{Amm-He-1998}
H.~Ammari, S.~He, An on-surface radiation condition for \textsc{M}axwell's
  equations in three dimensions, Microwave and Optical Technology Letters
  19~(1) (1998) 59--63.

\bibitem{Barucq-2002}
H.~Barucq, A new family of first-order boundary conditions for the
  \textsc{M}axwell system: derivation, well-posedness and long-time behavior,
  Journal de Math\'{e}matiques Pures et Appliqu\'{e}es 82~(1) (2002) 67--88.

\bibitem{Dom-Sayas-2013}
V.~Dominguez, F.-J. Sayas, Some properties of layer potentials and boundary
  integral operators for the wave equation, J. Integral Equations Appl. 25~(2)
  (2013) 253--294.

\end{thebibliography}

\end{document}